\documentclass[11pt]{article}
\usepackage{microtype}
\usepackage[latin1]{inputenc}
\usepackage{fullpage}   %
\usepackage{graphicx}  %
\usepackage{paralist} \plparsep=-0.1in \pltopsep=0.1in \plpartopsep=0.1in
\usepackage{enumitem}
\usepackage{complexity}
\usepackage{mleftright}
\usepackage{comment}
\usepackage{xspace}
\usepackage{nameref}
\usepackage[linktocpage=true,pagebackref=false]{hyperref}
\usepackage{cleveref}

\usepackage{amsmath}   %
\usepackage{amsfonts}   %
\usepackage{amssymb}   %
\usepackage{amsthm}     %
\usepackage{mathrsfs}

\newtheorem{theorem}{Theorem}

\newtheorem{lemma}{Lemma}

\usepackage{ifthen}

\newboolean{short}
\setboolean{short}{false}

\newcommand{\shortOnly}[1]{\ifthenelse{\boolean{short}}{#1}{}}
\newcommand{\onlyShort}[1]{\ifthenelse{\boolean{short}}{#1}{}}
\newcommand{\longOnly}[1]{\ifthenelse{\boolean{short}}{}{#1}}
\newcommand{\onlyLong}[1]{\ifthenelse{\boolean{short}}{}{#1}}

\def\polylog{\operatorname{polylog}}
\def\cA{\mathcal{A}}
\def\cC{\mathcal{C}}
\def\cM{\mathcal{M}}
\def\cP{\mathcal{P}}
\def\cT{\mathcal{T}}

\newcommand{\congest}{\ensuremath{\mathcal{CONGEST}}\xspace}

\newcommand{\avec}{\mathbf{a}}
\newcommand{\sketch}{\mathbf{s}}

\renewcommand{\paragraph}[1]{\medskip\noindent{\bf #1.}\xspace}

\newcommand*\samethanks[1][\value{footnote}]{\footnotemark[#1]}

\title{Fast Distributed Algorithms for Connectivity and MST in Large Graphs}

\author{Gopal Pandurangan\thanks{Department of Computer Science, University of Houston, Houston, TX 77204, USA.
 \hbox{E-mail}:~{\tt gopalpandurangan@gmail.com\textnormal{,} michele@cs.uh.edu}. Supported, in part, by US-Israel Binational Science Foundation grant 2008348, NSF grant CCF-1527867, and NSF grant CCF-1540512.}
 \and Peter Robinson\thanks{Department of Computer Science, Royal Holloway, University of London, United Kingdom.
 \hbox{E-mail}:~\texttt{peter.robinson@rhul.ac.uk}.}
 \and Michele Scquizzato\samethanks[1]
}

\begin{document}

\maketitle

\begin{abstract}
Motivated by the increasing need to understand the algorithmic foundations of distributed large-scale graph
computations, we study a number of fundamental graph problems in a message-passing model for distributed
computing where $k \geq 2$ machines jointly perform computations on graphs with $n$ nodes (typically, $n \gg k$).
The input graph is assumed to be initially randomly partitioned among the $k$ machines, a common
implementation in many real-world systems. Communication is point-to-point, and the goal is to minimize
the number of communication rounds of the computation.
 
Our main result is an (almost) optimal  distributed randomized algorithm for graph connectivity.
Our algorithm runs in $\tilde{O}(n/k^2)$ rounds ($\tilde{O}$ notation hides a $\polylog(n)$ factor and
an additive $\polylog(n)$ term). This improves over the best previously known bound of $\tilde{O}(n/k)$
[Klauck et al., SODA 2015], and is optimal (up to a polylogarithmic factor) in view of an existing lower
bound of $\tilde{\Omega}(n/k^2)$. Our improved algorithm uses a bunch of techniques, including linear graph
sketching, that prove useful in the design of efficient distributed graph algorithms.
Using the connectivity algorithm as a building block, we then present fast randomized algorithms
for computing minimum spanning trees, (approximate) min-cuts, and for many graph
verification problems. All these algorithms take $\tilde{O}(n/k^2)$ rounds, and are optimal up to polylogarithmic factors.
We also show an almost matching lower bound of $\tilde{\Omega}(n/k^2)$ rounds for many graph verification problems
by leveraging lower bounds in random-partition communication complexity.

\end{abstract}

\newpage

\section{Introduction}

The focus of this paper is on distributed computation on large-scale graphs, which is increasingly becoming important
with the rise of massive graphs such as the Web graph, social networks, biological networks, and other graph-structured data
and the consequent need for fast algorithms to process such graphs. Several large-scale graph processing systems such
as Pregel~\cite{pregel} and Giraph~\cite{giraph} have been recently designed based on the message-passing
distributed computing model \cite{Lynch96,Peleg00}. We study a number of fundamental graph problems in a
model which abstracts the essence of these graph-processing systems, and present almost tight bounds on the
time complexity needed to solve these problems.  %
In this model, introduced in~\cite{KlauckNPR15} and explained in detail in Section~\ref{sec:model},
the input graph is distributed across a group of $k \geq 2$ machines that are pairwise interconnected via a communication network. 
The $k$ machines jointly perform computations on an arbitrary $n$-vertex input graph, where typically $n \gg k$.
The input graph is assumed to be initially randomly partitioned among the $k$ machines (a common
implementation in many real world graph processing systems~\cite{pregel,Stanton14}). Communication
is point-to-point via message passing. The computation advances in synchronous rounds, and
there is a constraint on the amount of data that can cross each link of the network in each round. 
The goal is to minimize the time complexity, i.e., the number of rounds required by the computation.
This model is aimed at investigating the amount of ``speed-up'' possible vis-a-vis the number of available machines,
in the following sense: when $k$ machines are used, how does the time complexity scale in $k$? Which problems
admit linear scaling? Is it possible to achieve super-linear scaling?

Klauck et al.~\cite{KlauckNPR15} present lower and upper bounds for several fundamental graph
problems in the $k$-machine model. In particular, assuming that each link has a bandwidth of
one bit per round, they show a lower bound of $\tilde \Omega(n/k^2)$ rounds for the graph
connectivity problem.\footnote{Throughout this paper $\tilde O(f(n))$ denotes $O(f(n) \polylog n
+ \polylog n)$, and $\tilde \Omega(f(n))$ denotes $\Omega(f(n)/\polylog n)$.}
They also present an $\tilde{O}(n/k)$-round algorithm for graph connectivity and spanning tree (ST) verification.
This algorithm thus exhibits a scaling linear in the number of machines $k$. The question of existence of a faster algorithm,
and in particular of an algorithm matching the $\tilde \Omega(n/k^2)$ lower bound, was left open in~\cite{KlauckNPR15}.
In this paper we answer this question affirmatively by presenting an $\tilde O(n/k^2)$-round algorithm for graph connectivity,
thus achieving a speedup quadratic in $k$. This is optimal up to polylogarithmic (in $n$) factors.

This result is important for two reasons. First, it shows that there are non-trivial graph problems for
which we can obtain \emph{superlinear} (in $k$) speed-up. To elaborate further on this point, we shall
take a closer look at the proof of the lower bound for connectivity shown in~\cite{KlauckNPR15}.
Using communication complexity techniques, that proof shows that any (possibly randomized) algorithm
for the graph connectivity problem has to exchange $\tilde{\Omega}(n)$ bits of information across the $k$ machines,
for any $k \geq 2$. Since there are $k(k-1)/2$ links in a complete network with $k$ machines, when
each link can carry $O(\polylog(n))$ bits per round, in each single round the network can deliver at most
$\tilde{\Theta}(k^2)$ bits of information, and thus a lower bound of $\tilde{\Omega}(n/k^2)$ rounds follows.
The result of this paper thus shows that it is possible to exploit in full the available bandwidth, thus achieving
a speed-up of $\tilde{\Theta}(k^2)$. 
Second, this implies that many other important graph problems can be solved in $\tilde{O}(n/k^2)$ rounds as well.
These include computing a spanning tree, minimum spanning tree (MST), approximate min-cut, and many
verification problems such as spanning connected subgraph, cycle containment, and bipartiteness.

It is important to note that  under a different output requirement (explained next) 
there exists a $\tilde \Omega(n/k)$-round lower bound for computing a spanning tree of
a graph~\cite{KlauckNPR15}, which also implies the same lower bound for other fundamental problems
such as computing an MST, breadth-first tree, and shortest paths tree. 
However, this lower bound holds under the requirement that each vertex (i.e., the machine which hosts the vertex) must know
at the end of the computation the ``status'' of all of its incident edges, that is, whether they belong to an ST or not, and output
their respective status. (This is the output criterion that is usually required in distributed algorithms~\cite{Lynch96,Peleg00}.)
The proof of the lower bound exploits this criterion to show that any algorithm requires some machine receiving
$\Omega(n)$ bits of information, and since any machine has $k-1$ incident links, this results in a $\tilde \Omega(n/k)$ lower bound.
On the other hand, if we relax the output criterion to require the final status of each edge to be known by \emph{some}
machine, then we show that this can be accomplished in $\tilde{O}(n/k^2)$ rounds using the fast connectivity algorithm of this paper.

\subsection{The Model}\label{sec:model}

We now describe the adopted model of distributed computation, the \emph{$k$-machine model}
(a.k.a.\ the \emph{Big Data model}), introduced in~\cite{KlauckNPR15} and further investigated
in~\cite{fanchung,spidal,PanduranganRS16}. The model consists of a set of $k \geq 2$ machines $N = \{M_1,M_2,\dots,M_k\}$ 
that are pairwise interconnected by bidirectional point-to-point communication links.
Each machine executes an instance of a distributed algorithm. The computation advances
in synchronous rounds where, in each round, machines can exchange messages over their
communication links and perform some local computation. Each link is assumed to have
a bandwidth of $O(\polylog(n))$ bits per round, i.e., $O(\polylog(n))$ bits can be
transmitted over each link in each round. (As discussed in~\cite{KlauckNPR15} (cf.\ Theorem 4.1), it is easy to rewrite bounds to
scale in terms of the actual inter-machine bandwidth.)
 Machines do not share any memory and have no other
means of communication. There is an alternate (but equivalent) way to view this communication
restriction: instead of putting a bandwidth restriction on the links, we can put a restriction on
the amount of information that each \emph{machine} can communicate (i.e., send/receive)
in each round.  The results that we obtain in the bandwidth-restricted model will also apply
to the latter model~\cite{KlauckNPR15}.
Local computation within a machine is considered to happen instantaneously at zero cost, while
the exchange of messages between machines is the costly operation. (However, we note that in
all the algorithms of this paper, every machine in every round performs a computation bounded by a polynomial in $n$.)
We assume that each machine has access to a private source of true random bits.

Although the $k$-machine model is a fairly general model of computation, we are mostly interested in studying
graph problems in it. Specifically, we are given an input graph $G$ with $n$ vertices, each associated with a unique
integer ID from $[n]$, and $m$ edges. To avoid trivialities, we will assume that $n \geq k$ (typically, $n \gg k$).
Initially, the entire graph $G$  is not known by any single machine, but rather partitioned among the $k$ machines in a
``balanced'' fashion, i.e., the nodes and/or edges of $G$ are partitioned approximately evenly among the machines.
We assume a {\em  vertex-partition} model, whereby vertices, along with information of their incident edges,
are partitioned across machines. Specifically, the type of partition that we will assume throughout is the
{\em random vertex partition (RVP)}, that is, each vertex of the input graph is assigned randomly to one machine.
(This is the typical way used by many real systems, such as Pregel~\cite{pregel}, to partition the input graph
among the machines; it is easy to accomplish, e.g., via hashing.\footnote{In Section~\ref{sec:related} we will discuss an
alternate partitioning model, the \emph{random edge partition (REP)} model, where each edge of $G$ is
assigned independently and randomly to one of the $k$ machines, and show how the results in the random
vertex partition model can be related to the random edge partition model.})
However, we notice that our upper bounds also hold under the much weaker assumption whereby it is only
required that nodes and edges of the input graph are partitioned approximately evenly among the machines;
on the other hand, lower bounds under RVP clearly apply to worst-case partitions as well.

More formally, in the random vertex partition variant,  each vertex of $G$ is assigned independently and uniformly at
random to one of the $k$ machines. If a vertex $v$ is assigned to machine $M_i$ we say that $M_i$ is the {\em home machine}
of $v$ and, with a slight abuse of notation, write $v \in M_i$. When a vertex is assigned to a machine, {\em all its
incident edges} are assigned to that machine as well; i.e., the home machine will know the {IDs
of the neighbors of that vertex as well as the identity of the home machines of the neighboring vertices (and
the weights of the corresponding edges in case $G$ is weighted).
Note that an immediate property of the RVP model is that the number of vertices at each machine is \emph{balanced},
i.e., each machine is the home machine of $\tilde \Theta(n/k)$ vertices with high probability.
A convenient way to implement  the RVP model is through hashing: each vertex (ID)
is hashed to one of the $k$ machines. Hence, if a machine knows a vertex ID, it also knows where it is hashed to. 

Eventually, each machine $M_i$, for each $1 \leq i \leq k$, must set a designated local output variable $o_i$ (which need
not depend on the set of vertices assigned to $M_i$), and the \emph{output configuration} $o=\langle o_1,\dots,o_k\rangle$
must satisfy certain feasibility conditions for the problem at hand. For example, for the minimum spanning tree problem
each $o_i$ corresponds to a set of edges, and the edges in the union of such sets must form an MST of the input graph.

In this paper, we show results for distributed algorithms that are Monte Carlo.
Recall that a Monte Carlo algorithm is a randomized algorithm whose output may be incorrect with
some probability. Formally, we say that an algorithm computes a function $f$ with $\epsilon$-\emph{error}
if for every input it outputs the correct answer with probability at least $1 - \epsilon$, where
the probability is over the random partition and the random bit strings used by the algorithm (if any).
The \emph{round (time) complexity} of an algorithm is the maximum number of communication rounds until termination. 
For any $n$ and problem $\cP$ on $n$ node graphs, we let the {\em time complexity of solving $\cP$ with $\epsilon$ error probability} in the $k$-machine model, denoted by $\cT_\epsilon(\cP)$, be the minimum $T(n)$ such that there exists an $\epsilon$-error protocol that solves $\cP$ and terminates in $T(n)$ rounds. 
For any $0\leq \epsilon\leq 1$, graph problem $\cP$ and function $T:\mathbb{Z}_+\rightarrow \mathbb{Z}_+$, we say that $\cT_\epsilon(\cP)=O(T(n))$ if there exists integer $n_0$ and $c$ such that for all $n\geq n_0$, $\cT_\epsilon(\cP)\leq cT(n)$. Similarly, we say that $\cT_\epsilon(\cP)=\Omega(T(n))$ if there exists integer $n_0$ and real $c$ such that for all $n\geq n_0$, $\cT_\epsilon(\cP)\geq cT(n)$. For our upper bounds, we will usually use $\epsilon = 1/n$, which will imply high probability algorithms, i.e., succeeding with probability at least $1 - 1/n$. In this case, we will sometimes just omit $\epsilon$ and simply say 
the time bound applies ``with high probability.''

\subsection{Our Contributions and Techniques}
The main result of this paper, presented in Section~\ref{sec:algorithm}, is a randomized Monte Carlo algorithm in the $k$-machine model that determines
the connected components of an undirected graph $G$ correctly with high probability and that terminates in
$\tilde O(n/k^2)$ rounds.\footnote{Since the focus is on
the scaling of the time complexity with respect to $k$, we omit explicitly stating
the polylogarithmic factors in our run time bounds. However, the hidden polylogarithmic factor is not large---at most $O(\log^3 n)$.}
This improves upon the previous best bound of $\tilde{O}(n/k)$~\cite{KlauckNPR15},
since it is strictly superior in the wide range of parameter $k = \Theta(n^{\epsilon})$, for all constants $\epsilon \in (0,1)$.
Improving over this bound is non-trivial since various attempts to get a faster connectivity algorithm fail due to
the fact that they end up congesting a particular machine too much, i.e., up to $n$ bits may need to be sent/received
by a machine, leading to a $\tilde O(n/k)$ bound (as a machine has only $k-1$ links). For example, a simple algorithm for
connectivity is simply flooding: each vertex floods the lowest labeled vertex that it has seen so far; at the end each vertex will
have the label of the lowest labeled vertex in its component.\footnote{This algorithm has been implemented
in a variant of Giraph~\cite{TianBCTM13}.} It can be shown that the above algorithm
takes $\Theta(n/k + D)$ rounds (where $D$ is the graph diameter) in the $k$-machine model by using the
Conversion Theorem of~\cite{KlauckNPR15}. Hence new techniques are needed to break the $n/k$-round barrier.
 
Our connectivity algorithm is the result of the application of the following three techniques.

 \emph{1.\ Randomized Proxy Computation.} This technique, similar to known techniques used in randomized routing algorithms~\cite{Valiant82},
is used to load-balance congestion at any given machine by
redistributing it evenly across the $k$ machines. This is achieved, roughly speaking, by re-assigning the executions of individual nodes
uniformly at random among the machines. 
It is crucial to distribute the computation {\em and} communication across machines
to avoid congestion at any particular machine. In fact, this allows one to move away from the
communication pattern imposed by the topology of the input graph (which can cause congestion at a particular machine)
to a more balanced communication. 

 \emph{2.\ Distributed Random Ranking (DRR).} DRR~\cite{ChenP12} is a simple technique that will be used to build trees of low
height in the connectivity algorithm. Our connectivity algorithm is divided into phases, in each of which we do the following:
each current component (in the first phase, each vertex is a component by itself) chooses one {\em outgoing edge} and then components
are combined by merging them along outgoing edges.  If done naively, this may result in a long chain of merges, resulting in a component
tree of high diameter; communication along this tree will then take a long time. To avoid this we resort to DRR, which suitably reduces the
number of merges. With DRR, each component chooses a random {\em rank}, which is simply a random number, say in the interval $[1,n^3]$;
a component $C_i$ then merges with the component $C_j$ on the other side of its selected outgoing edge if and only if
the rank of $C_j$ is larger than the rank of $C_i$. Otherwise, $C_i$ does not merge with $C_j$, and thus
it becomes the root of a DRR tree, which is a tree induced by the components and the set of the outgoing edges that have been used
in the above merging procedure. It can be shown that the height of a DRR tree is bounded by $O(\log n)$ with high probability. 

 \emph{3.\ Linear Graph Sketching.} Linear graph sketching~\cite{AhnGM12a,AhnGM12b,McGregor14} is crucially helpful
in efficiently finding an outgoing edge of a component. A \emph{sketch} for a vertex (or a component)
is a short ($O(\polylog n)$) bit vector that efficiently encodes the adjacency list of the vertex. Sampling from this sketch
gives a random (outgoing) edge of this vertex (component). A very useful property is the linearity of the sketches:
adding the sketches of a set of vertices gives the sketch of the component obtained by combining the vertices; the edges between
the vertices (i.e., the intra-component edges) are automatically ``cancelled'', leaving only a sketch of the outgoing edges. Linear
graph sketches were originally used to process dynamic graphs in the (semi-) streaming model~\cite{AhnGM12a,AhnGM12b,McGregor14}. 
Here, in a distributed setting, we use them to reduce the amount of communication needed to find an outgoing edge;
in particular, graph sketches will avoid us from checking whether an edge is an inter-component or an intra-component edge,
and this will crucially reduce communication across machines. We note that earlier distributed algorithms such as the classical
GHS algorithm~\cite{GallagerHS83} for the MST problem would incur too much communication since they involve checking the
status of each edge of the graph.

We observe that it does not seem straightforward to effectively exploit these techniques in the $k$-machine model: for
example, linear sketches can be easily applied in the distributed streaming model by sending to a coordinator machine the
sketches of the partial stream, which then will be added to obtain the sketch of the entire stream. Mimicking
this trivial strategy in the $k$-machine model model would cause too much congestion at one node, leading
to a $\tilde O(n/k)$ time bound.

Using the above techniques and the fast connectivity algorithm, in Section~\ref{sec:applications} we give algorithms
for many other important graph problems. In particular, we present a $\tilde{O}(n/k^2)$-round algorithm for
computing an MST (and hence an ST). We also present $\tilde{O}(n/k^2)$-round algorithms for 
approximate min-cut, and for many graph verification problems including spanning connected subgraph, cycle containment, and bipartiteness.
All these algorithms are optimal up to a polylogarithmic factor.

In \Cref{sec:LBs} we show a lower bound of  $\tilde{\Omega}(n/k^2)$ rounds for many verification problems
by simulating the $k$-machine model in a $2$-party model of communication complexity where the inputs are randomly assigned to the players.

\subsection{Related Work}\label{sec:related}
The theoretical study of large-scale graph computations in distributed systems is relatively new. 
Several works have been devoted to developing MapReduce graph algorithms
(see, e.g., \cite{soda-mapreduce,filtering-spaa,LeskovecRU14} and references therein).  
We note that the flavor of the theory developed for MapReduce is quite different compared
to the one for the $k$-machine model. Minimizing communication is also the key goal in MapReduce algorithms; however this is
usually achieved by making sure that the data is made small enough quickly (that is, in a small number of MapReduce rounds)
to fit into the memory of a single machine (see, e.g., the MapReduce algorithm for MST in~\cite{filtering-spaa}).

For a  comparison of the $k$-machine model with other models for parallel and distributed processing, including Bulk-Synchronous Parallel (BSP) model~\cite{bsp},  MapReduce~\cite{soda-mapreduce}, and the congested clique, we refer to \cite{grigory-blog}. 
In particular, according to~\cite{grigory-blog},
``Among all models with restricted communication the ``big data'' [$k$-machine] model is the one most similar to the MapReduce model".  

The $k$-machine model is closely related to the BSP model; it can be considered to be a simplified version of BSP,
where the costs of local computation and of synchronization (which happens at the end of every round) are ignored.
Unlike the BSP and refinements thereof, which have several different parameters that make the analysis of algorithms
complicated~\cite{grigory-blog}, the $k$-machine model is characterized by just one parameter, the number of machines;
this makes the model simple enough to be analytically tractable, thus easing the job of designing and analyzing algorithms,
while at the same time it still captures the key features of large-scale distributed computations.

The $k$-machine model is related to the classical \congest model~\cite{Peleg00}, and in
particular to the \emph{congested clique} model, which recently has received considerable attention (see,
e.g., \cite{LotkerPPP05,LenzenW11,Lenzen13,DruckerKO14,Nanongkai14,Censor-HillelKKLPS15,HegemanPPSS15}).
The main difference is that the $k$-machine model is aimed at the study of large-scale computations,
where the size $n$ of the input is significantly bigger than the number of available machines $k$, and thus
many vertices of the input graph are mapped to the same machine, whereas the two aforementioned models are aimed
at the study of distributed network algorithms, where $n = k$ and thus each vertex corresponds to a dedicated machine.
More ``local knowledge'' is
available per vertex (since it can access for free information about other vertices in the same machine) in the $k$-machine
model compared to the other two models. On the other hand, all vertices assigned to a machine have to communicate
through the links incident on this machine, which can limit the bandwidth (unlike the other two models where
each vertex has a dedicated processor). These differences manifest in the time complexity. In particular, the fastest
known distributed algorithm in the congested clique model for a given problem may not give rise to the fastest
algorithm in the $k$-machine model. For example, the fastest algorithms for MST in the congested clique model
(\cite{LotkerPPP05,HegemanPPSS15}) require $\Theta(n^2)$ messages; implementing
these algorithms in the $k$-machine model requires $\Theta(n^2/k^2)$ rounds. Conversely, the slower GHS
algorithm~\cite{GallagerHS83} gives an  $\tilde{O}(n/k)$ bound in the $k$-machine model.
The recently developed techniques (see, e.g., \cite{DasSarmaHKKNPPW12,podc11,podc14,HegemanPPSS15,DruckerKO14})
used to prove time lower bounds in the standard \congest model and in the congested clique model are not
directly applicable here.

The work closest in spirit to ours is the recent work of Woodruff and Zhang~\cite{woodruff}.
This paper considers a number of basic statistical and graph problems in a distributed
message-passing model similar to the $k$-machine model.
However, there are some important differences. First, their model is asynchronous, and the cost function is the communication complexity,
which refers to the total number of bits exchanged by the machines during the computation.
Second,  a {\em worst-case} distribution of the input is assumed, while we assume a random distribution.
Third, which is an important difference, they assume an edge partition model for the problems on graphs, that is, the edges of the graph (as opposed to its vertices) are partitioned across the $k$ machines. In particular, for the connectivity problem, they show a
message complexity lower bound of $\tilde{\Omega}(nk)$ which essentially translates to a $\tilde{\Omega}(n/k)$ round lower bound in the $k$-machine model; it can be shown by using their proof technique   that this lower bound also applies to the {\em random edge partition (REP)} model, where edges are partitioned randomly among machines, as well. On the other hand, it is easy to show an
$\tilde{O}(n/k)$ upper bound for the connectivity in the REP model for connectivity and MST.\footnote{The high-level idea of the MST algorithm in the REP model is: (1) First ``filter'' the edges assigned to one machine using the cut and cycle properties of a MST~\cite{KargerKT95}; this leaves each machine with $O(n)$ edges;
(2) Convert this edge distribution to a RVP which can be accomplished in $\tilde{O}(n/k)$ rounds via hashing the vertices randomly to machines and then routing the edges appropriately; then apply the RVP bound.} Hence, in the REP model, $\tilde{\Theta}(n/k)$ is a tight bound for connectivity and other related problems such as MST.
However, in contrast, in the RVP model (arguably, a more natural partition model), we show that $\tilde{\Theta}(n/k^2)$ is the tight bound. Our results are a  step towards a better understanding
of the complexity of distributed graph computation vis-a-vis  the partition model.

From the technical point of view, King et al.~\cite{KingKT15} also use an idea similar to linear sketching.
Their technique might also be useful in the context of the $k$-machine model.

\section{The Connectivity Algorithm}\label{sec:algorithm}
In this section we present our main result, a Monte Carlo randomized algorithm for the $k$-machine
model that determines the connected components of an undirected graph $G$ correctly
with high probability and that terminates in $\tilde O(n/k^2)$ rounds with high probability.
This algorithm is optimal, up to $\polylog(n)$-factors, by virtue of a lower bound of
$\tilde\Omega(n/k^2)$ rounds~\cite{KlauckNPR15}.

Before delving into the details of our algorithm, as a warm-up we briefly discuss simpler, but less efficient, approaches. 
The easiest way to solve \emph{any} problem in our model is to first collect all available graph data at a single
machine and then solve the problem locally. For example, one could first elect a referee among the machines,
which requires $O(1)$ rounds~\cite{KuttenPPRT15}, and then instruct every machine to send its local data to the
referee machine. Since the referee machine needs to receive $O(m)$ information in total but has only $k-1$ links
of bounded bandwidth, this requires $\Omega(m/k)$ rounds. 

A more refined approach to obtain a distributed algorithm for the $k$-machine model is to use the Conversion
Theorem of~\cite{KlauckNPR15}, which provides a simulation of a congested clique algorithm $\cA$ in
$\tilde O(M/k^2 + \Delta' T /k)$ rounds in the $k$-machine model, where $M$ is the message complexity
of $\cA$, $T$ is its round complexity, and $\Delta'$ is an upper bound to the total number of messages
sent (or received) by a single node in a single round. (All these parameters refer to the performance of $\cA$ in the congested
clique model.) Unfortunately, existing algorithms (e.g., \cite{GallagerHS83,Thurimella97}) typically require
$\Delta'$ to scale to the maximum node degree, and thus the converted time complexity
bound in the $k$-machine model becomes $\tilde \Omega(n/k)$ at best.
Therefore, in order to break the $\tilde\Omega(n/k)$ barrier, we must develop new techniques that directly
exploit the additional locality available in the $k$-machine model. 

In the next subsection we give a high level overview of our algorithm, and then formally present all
the technical details in the subsequent subsections.

\subsection{Overview of the Algorithm}\label{sec:overview}
Our algorithm follows a Boruvka-style strategy~\cite{Boruvka26}, that is, it repeatedly merges
adjacent \emph{components} of the input graph $G$, which are connected subgraphs of $G$, to form
larger (connected) components. The output of each of these phases is a \emph{labeling} of the nodes of $G$
such that nodes that belong to the same current component have the same label. At the beginning of
the first phase, each node is labeled with its own unique ID, forms a distinct component, and is also the
\emph{component proxy} of its own component. Note that, at any phase, a component contains up
to $n$ nodes, which might be spread across different machines; we use the term \emph{component part}
to refer to all those nodes of the component that are held by the same machine. Hence, at any phase
every component is partitioned in at most $k$ component parts. At the end of the algorithm each vertex
has a label such that any two vertices have the same label if and only if they belong to the same connected
component of $G$.

Our algorithm relies on \emph{linear graph sketches} as a tool to enable communication-efficient
merging of multiple components. Intuitively speaking, a (random) linear sketch $\sketch_u$ of a node $u$'s
graph neighborhood returns a sample chosen uniformly at random from $u$'s incident edges.
Interestingly, such a linear sketch can be represented as matrices using only $O(\polylog(n))$
bits~\cite{JowhariST11,McGregor14}. A crucial property of these sketches is that they are linear:
that is, given sketches $\sketch_u$ and $\sketch_v$, the \emph{combined sketch} $\sketch_u + \sketch_v$
(``+'' refers to matrix addition) has the property that, w.h.p., it yields a random sample of the edges incident to
$(u,v)$ in a graph where we have contracted the edge $(u,v)$ to a single node. We describe the technical details
in \Cref{sec:sketches}.

We now describe how to communicate these graph sketches in an efficient manner:
Consider a component $C$ that is split into $j$ parts $P_1,P_2,\dots,P_j$, the nodes of which
are hosted at machines $M_1,M_2,\dots,M_j$. To find an outgoing edge for $C$, we first instruct each
machine $M_i$ to construct a linear sketch of the graph neighborhood of each of the nodes in part $P_i$.
Then, we sum up these $|P_i|$ sketches, yielding a sketch $\sketch_{P_i}$ for the neighborhood of part $P_i$. 
To combine the sketches of the $j$ distinct parts, we now select a \emph{random component proxy}
machine $M_{C,r}$ for the current component $C$ at round $r$ (see Section~\ref{sec:proxy}).
Next, machine $M_i$ sends $\sketch_{P_i}$ to machine $M_{C,r}$; note that this causes at most $k$ messages
to be sent to the component proxy. Finally, machine $M_{C,r}$ computes $\sketch_C = \sum_{i=1}^j \sketch_{P_i}$,
and then uses $\sketch_C$ to sample an edge incident to some node in $C$, which, by construction, is guaranteed
to have its endpoint in a distinct component $C'$. (See Section~\ref{sec:edge}.)

At this point, each component proxy has sampled an inter-component edge
inducing the edges of a \emph{component graph} $\cC$ where each vertex corresponds to a component.
To enable the efficient merging of components, we employ the \emph{distributed random ranking} (DRR) technique of~\cite{ChenP12}
to break up any long paths of $\cC$ into more manageable directed trees of depth $O(\log n)$.
To this end, every component chooses a rank independently
and uniformly at random from $[0,1]$,\footnote{It is easy to see that an accuracy of $\Theta(\log n)$ bits suffices to break ties w.h.p.} and each component (virtually) connects to its neighboring component (according to $\cC$) via a (conceptual) directed edge if and only if the latter has a higher rank. 
Thus, this process results in a collection of disjoint rooted trees, rooted at the node of highest (local) rank.
We show in \Cref{sec:drr} that each of such trees has depth $O(\log n)$.

The merging of the components of each tree $\cT$ proceeds from the leafs upward (in parallel for each tree).
In the first merging phase, each leaf $C_j$ of $\cT$ merges with its parent $C'$ by relabeling the component
labels of all of their nodes with the label of $C'$. Note that the proxy $M_{C_j}$ knows the labeling of $C'$,
as it has computed the outgoing edge from a vertex in $C_j$ to a vertex in $C'$.
Therefore, machine $M_{C_j}$ sends the label of $C_j$ to all the machines that hold a part of $C_j$.
In~\Cref{sec:drr} we show that this can be done in parallel (for all leafs of all trees) in $\tilde O(n/k^2)$ rounds.
Repeating this merging procedure $O(\log n)$ times, guarantees that each tree has been merged to a single component. 

Finally, in Section~\ref{sec:time} we prove that $O(\log n)$ repetitions of the above process suffice to ensure that
the components at the end of the last phase correspond to the connected components of the input graph $G$.

\subsection{Communication via Random Proxy Machines}\label{sec:proxy}
Recall that our algorithm iteratively groups vertices into components and subsequently merges such components
according to the topology of $G$. Each of these components may be split into multiple component parts spanning
multiple machines. Hence, to ensure efficient load balancing of the messages that machines need to send on
behalf of the component parts that they hold, the algorithm performs all communication via \emph{proxy machines}.

Our algorithm proceeds in phases, and each phase consists of iterations. 
Consider the $\rho$-iteration of the $j$-th phase of the algorithm, with $\rho,j \ge 1$.
We construct a ``sufficiently'' random hash function $h_{j,\rho}$, such that, for each component
$C$, the machine with ID $h_{j,\rho}(C) \in [k]$ is selected as the proxy machine for component $C$.
First, machine $M_1$ generates $\ell = \tilde \Theta(n/k)$ random bits from its private source of randomness.
$M_1$ will distribute these random bits to all other machines via the following simple routing mechanism that proceeds in sequences of two rounds.
$M_1$ selects $k$ bits $b_1,b_2,\dots,b_{k-1}$ from the set of its $\ell$ private random bits that remain to be distributed,
and sends bit $b_i$ across its $i$-th link to machine $M_{i+1}$.
Upon receiving $b_i$, machine $M_{i+1}$ broadcasts $b_i$ to all machines in the next round.
This ensures that bits $b_1,b_2,\dots,b_{k-1}$ become common knowledge within two rounds.
Repeating this process to distribute all the $\ell = \tilde \Theta(n/k)$ bits takes $\tilde O(n/k^2)$ rounds,
after those all the machines have the $\ell$ random bits generated by $M_1$.
We leverage a result of~\cite{alonBabai} (cf.\ in its formulation as Theorem~2.1 in~\cite{alon12}),
which tells us that we can generate a random hash function such that it is $d$-wise independent by using only $O(d \log n)$ true random bits. 
We instruct machine $M_1$ to disseminate $d=\ell\log n = n\polylog(n)/k$ of its random bits according
to the above routing process and then each machine locally constructs the same hash function $h_{j,\rho}$,
which is then used to determine the component proxies throughout iteration $\rho$ of phase $j$.

We now show that communication via such proxy machines is fast in the $k$-machine model.

\begin{lemma}\label{lem:proxyComm}
Suppose that each machine $M$ generates a message of size $O(\polylog(n))$ bits for each component part residing
on $M$; let $m_i$ denote the message of part $P_i$ and let $C$ be the component of which $P_i$ is a part.
If each $m_i$ is addressed to the proxy machine $M_C$ of component $C$, then all messages are delivered
within $\tilde O(n/k^2)$ rounds with high probability. 
\end{lemma}

\begin{proof}
Observe that, except for the very first phase of the algorithm, the claim does not immediately follow from
a standard balls-into-bins argument because not all the destinations of the messages are chosen independently and
uniformly at random, as any two distinct messages of the same component have the same destination.

Let us stipulate that any component part held by machine $M_i$ is the $i$-th component part of its component,
and denote this part with $P_{i,j}$, $i \in [k]$, $j \in [n]$, where $P_{i,j} = \emptyset$ means that in machine
$i$ there is no component part for component $j$. Suppose that the algorithm is in phase $j'$ and iteration $\rho$.
By construction, the hash function $h_{j',\rho}$ is $\tilde \Theta(n/k)$-wise independent, and all the component
parts held by a single machine are parts of different components.
Since $M_i$ has at most $\tilde \Theta(n/k)$ distinct component parts w.h.p., it follows that all the proxy
machines selected by the component parts held by machine $M_i$ are distributed independently and uniformly at random. 
Let $y$ be the number of distinct component parts held by a machine $M_i$ that is, $y = \vert\{P_{i,j} :  P_{i,j} \neq \emptyset\}\vert = \tilde O(n/k)$ (w.h.p.).

Consider a link of $M_i$ connecting it to another machine $M_1$. Let $X_t$ be the indicator variable that
takes value $1$ if $M_1$ is the component proxy of part $t$ (of $M_i$), and let $X_t=0$ otherwise.
Let $X = \sum_{i=1}^{y} X_i$ be the number of component parts that chose their proxy machine at the
endpoint of link $(M_i,M_1)$. Since $\text{Pr}(X_i = 1) = 1/(k-1)$, we have that
the expected number of messages that have to be sent by this machine over any specific link is
$\text{E}[X] = y/(k-1)$.

First, consider the case $y \geq 11k \log n$. As the $X_i$'s are $\tilde \Theta(n/k)$-wise independent,
all proxies by the component parts of $M_i$ are chosen independently and thus we can apply
a standard Chernoff bound (see, e.g., \cite{MitzenmacherU05}), which gives
\[
\text{Pr}\left(X \geq \frac{7y}{4(k-1)}\right) \leq e^{- 3y / 16(k-1)} < e^{\frac{-2k \log n}{k}} < \frac{1}{n^2}.
\]
By applying the union bound over the $k \leq n$ machines we conclude that w.h.p.\ every machine sends
$\tilde O(n/k^2)$ messages to each proxy machine, and this requires $\tilde O(n/k^2)$ rounds.

Consider now the case $y < 11k \log n$. It holds that $6\text{E}[X] = 6y/(k-1) < 6 \cdot 11k \log n /(k-1) \leq 132 \log n$,
and thus, by a standard Chernoff bound,
\[
\text{Pr}\left(X \geq 132 \log n\right) \leq 2^{-132 \log n} = \frac{1}{n^{132}}.
\]
Analogously to the first case, applying the union bound over the $k \le n$ machines yields the result.
\end{proof}

\subsection{Linear Graph Sketches}\label{sec:sketches}
As we will see in \Cref{sec:drr}, our algorithm proceeds by merging components across randomly chosen inter-component edges.
In this subsection we show how to provide these sampling capabilities in a communication-efficient way in the $k$-machine model
by implementing random linear graph sketches. Our description follows the notation of~\cite{McGregor14}.

Recall that each vertex $u$ of $G$ is associated with a unique integer ID from $[n]$ (known to its home machine)
which, for simplicity, we also denote by $u$.\footnote{Note that the asymptotics of our results do not change
if the size of the ID space is $O(\polylog(n))$.}
For each vertex $u$ we define the \emph{incidence vector $\avec_u \in \{-1,0,1\}^{n\choose 2}$ of $u$},
which describes the incident edges of $u$, as follows:
\begin{align*}
\avec_u[(x,y)] = \left\{
 \begin{array}{ll}
 1& \,\text{if $u = x < y$ and $(x,y) \in E$,}\\
 -1& \,\text{if $x < y = u$ and $(x,y) \in E$,}\\
 0& \,\text{otherwise.}\\
 \end{array}
 \right.
\end{align*}
Note that the vector $\avec_u + \avec_v$ corresponds to the incidence vector of the contracted edge $(u,v)$.
Intuitively speaking, summing up incidence vectors ``zeroes out'' edges between the corresponding vertices,
hence the vector $\sum_{u \in C} \avec_u$ represents the outgoing edges of a component $C$.

Since each incidence vector $\avec_u$ requires polynomial space, it would be inefficient to directly communicate
vectors to component proxies. Instead, we construct a random linear sketch $\sketch_u$ of $\polylog(n)$-size
that has the property of allowing us to sample uniformly at random a nonzero entry of $\avec_u$ (i.e., an edge incident
to $u$). (This is referred to as $\ell_0$-sampling in the streaming literature, see e.g.\ \cite{McGregor14}.)
It is shown in~\cite{JowhariST11} that $\ell_0$-sampling can be performed by linear projections.
Therefore, at the beginning of each phase $j$ of our algorithm, we instruct each machine to to
create a new (common) $\polylog(n) \times {n\choose 2}$ sketch matrix $L_j$, which we call \emph{phase $j$ sketch matrix}.\footnote{Here we describe the construction as if nodes have access to a source of shared randomness (to create the sketch matrix).\onlyLong{ We later show how to remove this assumption.}\onlyShort{ In the full paper we show how to remove this assumption.}}
Then, each machine $M$ creates a sketch $\sketch_u = L_j \cdot \avec_u$ for each vertex $u$ that resides on $M$.
Hence, each $\sketch_u$ can be represented by a polylogarithmic number of bits.

Observe that, by linearity, we have $L_j \cdot \avec_u + L_j \cdot \avec_v = L_j \cdot (\avec_u + \avec_v)$.
In other words, a crucial property of sketches is that the sum $\sketch_u + \sketch_v$ is itself a sketch that
allows us to sample an edge incident to the contracted edge $(u,v)$. We summarize these properties
in the following statement.

\begin{lemma} \label{lem:sketches}
Consider a phase $j$, and let $P$ a subgraph of $G$ induced by vertices $\{u_1,\dots,u_\ell\}$.
Let $\sketch_{u_1},\dots,\sketch_{u_\ell}$ be the associated sketches of vertices in $P$ constructed
by applying the \emph{phase $j$ sketch matrix} to the respective incidence vectors.
Then, the combined sketch $\sketch_P = \sum_{i=1}^\ell \sketch_{u_i}$  can be represented using
$O(\polylog(n))$ bits and, by querying $\sketch_P$, it is possible (w.h.p.) to sample a random edge
incident to $P$ (in $G$) that has its other endpoint in $G\setminus P$.
\end{lemma}

\paragraph{Constructing Linear Sketches Without Shared Randomness} 
Our construction of the linear sketches described so far requires $\tilde O(n)$ fully independent random bits that would need to be shared by all machines.
It is shown in Theorem~1 (cf.\ also Corollary~1) of~\cite{CormodeF14} that it is possible to construct such an
$\ell_0$-sampler (having the same linearity properties) by using $\Theta(n)$ random bits that are only $\Theta(\log n)$-wise independent. 
Analogously as in \Cref{sec:proxy}, we can generate the required $\Theta(\log^2 n)$ true random bits at machine $M_1$,
distribute them among all other machines in $O(1)$ rounds, and then invoke Theorem~2.1 of~\cite{alon12}
at each machine in parallel to generate the required (shared) $\Theta(\log n)$-wise independent random bits for constructing the sketches.

\subsection{Outgoing Edge Selection}\label{sec:edge}
Now that we know how to construct a sketch of the graph neighborhood of any set of vertices,
we will describe how to combine these sketches in a communication-efficient way in the $k$-machine model.
The goal of this step is, for each (current) component $C$, to find an outgoing edge that connects $C$
to some other component $C'$.

Recall that $C$ itself might be split into parts $P_1,P_2,\dots,P_j$ across multiple machines.
Therefore, as a first step, each machine $M_i$ locally constructs the combined sketch for each part that resides in $M_i$.
By \Cref{lem:sketches}, the resulting sketches have polylogarithmic size each and present a sketch of the incidences
of their respective component parts. Next, we combine the sketches of the individual parts of each component $C$
to a sketch of $C$, by instructing the machines to send the sketch of each part $P_i$ (of component $C$) to the proxy
machine of $C$. By virtue of \Cref{lem:proxyComm}, all of these messages are delivered to the component proxies
within $\tilde O(n/k^2)$ rounds. Finally, the component proxy machine of $C$ combines the received sketches to
yield a sketch of $C$, and randomly samples an outgoing edge of $C$ (see \Cref{lem:sketches}).
Thus, at the end of this procedure, every component (randomly) selected exactly one neighboring component.
We now show that the complexity of this procedure is $\tilde O(n/k^2)$ w.h.p.

\begin{lemma}\label{lem:outgoing}
Every component can select exactly one outgoing edge in $\tilde O(n/k^2)$ rounds with high probability.
\end{lemma}

\begin{proof}
Clearly, since at every moment each node has a unique component's label, each machine
holds $\tilde O(n/k)$ component's parts w.h.p. Each of these parts selected at most one edge,
and thus each machine ``selected'' $\tilde O(n/k)$ edges w.h.p. All these edges have to be sent
to the corresponding proxy. By Lemma~\ref{lem:proxyComm}, this requires $\tilde O(n/k^2)$ rounds.

The procedure is completed when the proxies communicate the decision to each
of the at most $k$ components' parts. This entails as many messages as in the first part to be
routed using exactly the same machines' links used in the first part, with the only difference
being that messages now travel in the opposite direction. The lemma follows.
\end{proof}

\subsection{Merging of Components}\label{sec:drr}
After the proxy machine of each component $C$ has selected one edge connecting $C$ to a different
component, all the neighboring components are merged so as to become a new, bigger component.
This is accomplished by relabeling the nodes of the graph such that all the nodes in the same (new)
component have the same label. Notice that the merging is thus only virtual, that is, component parts
that compose a new component are not moved to a common machine; rather, nodes (and their incident edges)
remain in their home machine, and just get (possibly) assigned a new label.

We can think of the components along with the sampled outgoing edges as a \emph{component graph} $\cC$.
We use the \emph{distributed random ranking} (DRR) technique~\cite{ChenP12} to avoid having long chains
of components (i.e., long paths in $\cC$). That is, we will (conceptually) construct a forest of directed trees
that is a subgraph (modulo edge directions) of the component graph $\cC$ and where each tree has depth $O(\log n)$.\footnote{Instead of using DRR trees, an alternate and simpler idea is the following. Let every component select a number in $[0,1]$. A merging can be done only if the outgoing edge (obtained from the sketch) connects a component with ID 0 to a component with ID 1. One can show that this merging procedure
also gives the same time bound.} %
The component proxy of each component $C$ chooses a rank independently and uniformly at random from $[0,1]$.
(It is easy to show that $\Theta(\log n)$ bits provide sufficient accuracy to break ties w.h.p.)
Now, the proxy machine of $C$ (virtually) connects $C$ to its neighboring component $C'$ if and only if
the rank chosen by the latter's proxy is higher. In this case, we say that $C'$ becomes the \emph{parent}
of $C$ and $C$ is a \emph{child} of $C'$.

\begin{lemma} \label{lem:drrComm}
After $\tilde O(n/k^2)$ rounds, the structure of the DRR-tree is completed with high probability.
\end{lemma}

\begin{proof}
We need to show that every proxy machine of a non-root component knows its smaller-ranking parent
component and every root proxy machine knows that it is root. Note that during this step the proxy
machines of the child components communicate with the respective parent proxy machines.
Moreover, the number of messages sent for determining the ordering of the DRR-trees is guaranteed
to be $O(n)$ with high probability, since $\cC$ has only $O(n)$ edges. 
By instantiating \Cref{lem:proxyComm}, it follows that the delivery of these messages can be completed in $\tilde O(n/k^2)$ rounds w.h.p.

Since links are bidirectional, the parent proxies are able to send their replies
within the same number of rounds, by re-running the message schedule of the child-to-parent communication in reverse order.
\end{proof}

If a component has the highest rank among all its neighbors (in $\cC$), we call it a \emph{root component}.
Since every component except root components connects to a component with higher rank, the resulting structure 
is a set of disjoint rooted trees.

In the next step, we will merge all components of each tree into a single new component such that all
vertices that are part of some component in this tree receive the label of the root. Consider a tree $\cT$.
We proceed level-wise (in parallel for all trees) and start the merging of components at the leafs that
are connected to a (lower-ranking) parent component $C$.

\begin{lemma}\label{lem:merging}
There is a distributed algorithm that merges all trees of the DRR forest in $\tilde O(d n/k^2)$ rounds with high probability,
where $d$ is the largest depth of any tree. 
\end{lemma}

\begin{proof}
We proceed in $d$ iterations by merging the (current) leaf components with their parents in the tree. 
Thus it is sufficient to analyze the time complexity of a single iteration.
To this end, we describe a procedure that changes the component labels of all vertices that are in leaf
components in the DRR forest to the label of the respective parent in $\tilde O(n/k^2)$ rounds.

At the beginning of each iteration, we select a new proxy for each component $C$ by querying the shared
hash function $h_{j,\rho}(C)$, where $\rho$ is the current iteration number.
This ensures that there are no dependencies between the proxies used in each iteration.
We know from~\Cref{lem:drrComm} that there is a message schedule such that leaf proxies
can communicate with their respective parent proxy in $\tilde O(n/k^2)$ rounds (w.h.p.) and vice versa,
and thus every leaf proxy knows the component label of its parent.
We have already shown in \Cref{lem:outgoing} that we can deliver a message from each component part
to its respective proxy (when combining the sketches) in $\tilde O(n/k^2)$ rounds.
Hence, by re-running this message schedule, we can broadcast the parent label from the leaf proxy to each component part in the same time.
Each machine that receives the parent label locally changes the component label of the vertices that are in the corresponding part. 
\end{proof}

The following result is proved in~\cite[Theorem~11]{ChenP12}. To keep the paper self-contained
we also provide a direct and simpler proof for this result (see Appendix).

\begin{lemma}[{\cite[Theorem~11]{ChenP12}}]\label{lem:depthDRR}
The depth of each DRR tree is $O(\log n)$ with high probability.
\end{lemma}

\subsection{Analysis of the Time Complexity}\label{sec:time}
We now show that the number of phases required by the algorithm to determine the connected
components of the input graph is $O(\log n)$.
At the beginning of each phase $i$, distributed
across the $k$ machines there are $c_i$ distinct components. At the beginning of the algorithm
each node is identified as a component, and thus $c_0 = n$. The algorithm ends at the completion
of phase $\varphi$, where $\varphi$ is the smallest integer such that $c_\varphi = cc(G)$, where $cc(G)$
denotes the number of connected components of the input graph $G$. If pairs of components
were merged in each phase, it would be straightforward to show that the process would terminate in
at most $O(\log n)$ phases. However, in our algorithm each component connects to its neighboring
component if and only if the latter has a higher rank. Nevertheless, it is not difficult to show
that this slightly different process also terminates in $O(\log n)$ phases w.h.p. (that is, components
gets merged ``often enough''). The intuition for this result is that, since components' ranks are
taken randomly, for each component the probability that its neighboring component has a higher
rank is exactly one half. Hence, on average half of the components will not be merged with their
own neighbor: each of these components thus becomes a root of one component, which means
that, on average, the number of new components will be half as well.

\begin{lemma}\label{lem:phases}
After $12 \log n$ phases, the component labels of the vertices correspond to the connected components of $G$ with high probability.
\end{lemma}

\begin{proof}
Replace the $c_i$'s with corresponding random variables $C_i$'s, and consider the stochastic
process defined by the sequence $C_0,C_1,\dots,C_\varphi$. Let $\bar C_i$ be the random
variable that counts the number of components that actually participate at the merging process
of phase $i$, because they do have an outgoing edge to another component. Call these components
\emph{participating components}. %
Clearly, by definition, $\bar C_i \leq C_i$.

We now show that, for every phase $i \in [\varphi-1]$, $\text{E}[\text{E}[\bar C_{i+1} \mid \bar C_i]]
\leq \text{E}[\bar C_{i}]/2$.
To this end, fix a generic phase $i$ and a random ordering of its $\bar C_i$ participating components.
Define random variables $X_{i,1},X_{i,2},\dots,X_{i,\bar C_i}$ where $X_{i,j}$ takes value $1$ if
the $j$-th participating component will be a root of a participating tree/component for phase $i+1$,
and $0$ otherwise. Then, $\bar C_{i+1} \mid \bar C_i = \sum_{j=1}^{\bar C_i} X_{i,j}$ is the number
of participating components for phase $i+1$. As we noticed before, for any $i \in [\varphi-1]$ and
$j \in [\bar C_i]$, the probability that a participating component will not be merged to its neighboring
component, and thus become a root of a tree/component for phase $i+1$ is exactly one half.
Therefore,
\[
\text{Pr}(X_{i,j} = 1) \leq 1/2.
\]
Hence, by the linearity of expectation, we have that
\[
\text{E}[\bar C_{i+1} \mid \bar C_i] =
\sum_{j=1}^{\bar C_i}\text{E}[X_{i,j}] =
\sum_{j=1}^{\bar C_i}\text{Pr}(X_{i,j} = 1)
\leq \frac{\bar C_i}{2}.
\]
Then, using again the linearity of expectation,
\[
\text{E}[\text{E}[\bar C_{i+1} \mid \bar C_i]] \leq
\text{E}\mleft[\frac{\bar C_i}{2}\mright] =
\frac{\text{E}[\bar C_{i}]}{2}.
\]

We now leverage this result to prove the claimed statement. Let us call a phase \emph{successful}
if it reduces the number of participating components by a factor of at most $3/4$. By Markov's
inequality, the probability that phase $i$ is not successful is
\begin{align*}
\text{Pr}\mleft(\text{E}[\bar C_{i+1} \mid \bar C_i] > \frac{3}{4} \text{E}[\bar C_i]\mright) &<
\frac{\text{E}[\text{E}[\bar C_{i+1} \mid \bar C_i] ]}{(3/4) \text{E}[\bar C_i]} \\
&\leq \frac{\text{E}[\bar C_i]}{2} \cdot \frac{4}{3\text{E}[\bar C_i]} \\
&= \frac{2}{3},
\end{align*}
and thus the probability that a phase of the algorithm is successful is at least $1/3$. Now consider
a sequence of $12 \log n$ phases of the algorithm. We shall prove that within that many phases the
algorithm w.h.p.\ has reduced the number of participating components a sufficient number of times
so that the algorithm has terminated, that is, $\varphi \leq 12 \log n$ w.h.p. Let $X_i$ be an indicator
variable that takes value $1$ if phase $i$ is successful, and $0$ otherwise (this also includes the
case that the $i$-th phase does not take place because the algorithm already terminated). Let
$X = \sum_{i=1}^{12 \log n} X_i$ be the number of successful phases out of the at most
$12 \log n$ phases of the algorithm. Since $\text{Pr}(X_i = 1) \geq 1/3$,
by the linearity of expectation we have that
\[
\text{E}[X] = \sum_{i=1}^{12 \log n} \text{E}[X_i] =
\sum_{i=1}^{12 \log n} \text{Pr}(X_i = 1) \geq
\frac{12 \log n}{3} =
4 \log n.
\]
As the $X_i$'s are independent we can apply a standard Chernoff bound, which gives
\[
\text{Pr}(X \leq \log n) \leq  e^{-4 \log n (3/4)^2 / 2} = e^{-\frac{9}{8}\log n} < \frac{1}{n}.
\]
Hence, with high probability $12 \log n$ phases are enough to determine all the components
of the input graph.
\end{proof}

\begin{theorem}\label{thm:connectivity}
There is a distributed algorithm in the $k$-machine model that determines the connected
components of a graph $G$ in $\tilde O(n/k^2)$ rounds with high probability.
\end{theorem}

\begin{proof}
By \Cref{lem:phases}, the algorithm finishes in $O(\log n)$ phases with high probability.
To analyze the time complexity of an individual phase, recall that it takes $\tilde O(n/k^2)$ rounds to sample an outgoing edge (see \Cref{lem:outgoing}).
Then, building the DRR forest requires $\tilde O(n/k^2)$ additional rounds, according to~\Cref{lem:drrComm}.
Merging each DRR tree $\cT$ in a level-wise fashion (in parallel) takes $\tilde O(d n/k^2)$ rounds
(see \Cref{lem:merging}), where $d$ is the depth of $\cT$ which, by virtue of \Cref{lem:depthDRR}, is bounded by $O(\log n)$.
Since each of these time bounds hold with high probability, and the algorithm consists of $O(\log n)$ phases with high probability,
by the union bound we conclude that the total time complexity of the algorithm is $\tilde O(n/k^2)$ with high probability.
\end{proof}

We conclude the section by noticing that it is easy to output the actual number of connected components
after the termination of our algorithm: every machine just needs to send ``YES'' directly to the proxies of
each of the components' labels it holds, and subsequently such proxies will send the labels of the components
for which they received ``YES'' to one predetermined machine. Since the communication is performed via
the components' proxies, it follows from Lemma~\ref{lem:proxyComm} that the first step takes $\tilde O(n/k^2)$
rounds w.h.p., and the second step takes only $O(\log n)$ rounds w.h.p.

\section{Applications}\label{sec:applications}

In this section we describe how to use our fast connectivity algorithm as a building block to solve
several other fundamental graph problems in the $k$-machine model in time $\tilde O(n/k^2)$.

\subsection{Constructing a Minimum Spanning Tree}
Given a weighted graph where each edge $e=(u,v)$ has an associated weight $w(e)$, initially
known to both the home machines of $u$ and $v$, the minimum spanning tree (MST) problem
asks to output a set of edges that form a tree, connect all nodes, and have the minimum possible
total weight. Klauck et al.~\cite{KlauckNPR15} show that $\tilde \Omega(n/k)$ rounds are
necessary for constructing any spanning tree (ST), assuming that, for every spanning tree
edge $e=(u,v)$, the home machine of $u$ and the home machine of $v$ must \emph{both}
output $(u,v)$ as being part of the ST. Here we show that we can break the $\tilde \Omega(n/k)$
barrier, under the slightly less stringent requirement that each spanning tree edge $e=(u,v)$
is returned by at least \emph{one} machine, but not necessarily by both the home machines of $u$ and $v$.

Our algorithm mimics the multi-pass MST construction procedure of~\cite{AhnGM12a},
originally devised for the (centralized) streaming model.
To this end we modify our connectivity procedure of \Cref{sec:algorithm}, by ensuring that
when a component proxy $C$ chooses an outgoing edge $e$, this is the \emph{minimum weight
outgoing edge} (MWOE) of $C$ with high probability.

We now describe the $i$-th phase of this MST construction in more detail.
Analogously to the connectivity algorithm in \Cref{sec:algorithm}, the proxy of each component $C$
determines an outgoing edge $e_0$ which, by the guarantees of our sketch construction
(\Cref{lem:sketches}), is chosen uniformly at random from all possible outgoing edges of $C$.

We then repeat the following edge-elimination process $t=\Theta(\log n)$ times:
The proxy broadcasts $w(e_0)$ to every component part of $C$.
Recall from \Cref{lem:outgoing} that this communication is possible in $\tilde O(n/k^2)$ rounds.
Upon receiving this message, the machine $M$ of a part $P$ of $C$ constructs a new sketch
$\sketch_u$ for each $u \in P$, but first zeroes out all entries in $\avec_u$ that refer to edges of weight $>w(e_0)$.
(See \Cref{sec:sketches} for a more detailed description of $\avec_u$ and $\sketch_u$.)
Again, we combine the sketches of all vertices of all parts of $C$ at the proxy of $C$, which in turn samples a new outgoing edge $e_1$ for $C$.
Since each time we sample a randomly chosen edge and eliminate all higher weight edges, it is easy to see that the edge $e_t$  is the MWOE of $C$ w.h.p.
Thus, the proxy machine of $C$ includes the edge $e_t$ as part of the MST output.
Note that this additional elimination procedure incurs only a logarithmic time complexity overhead.

At the end of each phase, we  proceed by (virtually) merging the components along their MWOEs
in a similar manner as for the connectivity algorithm (see \Cref{sec:drr}), thus requiring $\tilde O(n/k^2)$ rounds in total.

Let $E$ be the set of added outgoing edges.
Since the components of the connectivity algorithm eventually match the actual components of the input
graph, the graph $H$ on the vertices $V(G)$ induced by $E$ connects all vertices of $G$.
Moreover, since components are merged according to the trees of the DRR-process (see \Cref{sec:drr}), it follows that $H$ is cycle-free.

%
We can now fully classify the complexity of the MST problem in the $k$-machine model:

\begin{theorem} \label{thm:mst}
There exists an algorithm for the $k$-machine model that outputs an MST in
\begin{enumerate}[label=(\alph*)]
\item $\tilde O(n/k^2)$ rounds, if each MST-edge is output by at least one machine, or in
\item $\tilde O(n/k)$ rounds, if each MST-edge $e$ is output by both machines that hold an endpoint of $e$.
\end{enumerate}
Both bounds are tight up to polylogarithmic factors.
\end{theorem}

\subsection{{\large\boldmath$O(\log n)$}-Approximation for Min-Cut}

Here we show the following result for the min-cut problem in the $k$-machine model.

\begin{theorem}
There exists an $O(\log n)$-approximation algorithm for the min-cut problem in the $k$-machine model
that runs in $\tilde O(n/k^2)$ rounds with high probability.
\end{theorem}

\begin{proof}
We use exponentially growing sampling probabilities for sampling edges and then check
connectivity, leveraging a result by Karger~\cite{Karger94}. This procedure was proposed
in~\cite{GhaffariK13} in the classic $\mathcal{CONGEST}$ model, and can be implemented in
the $k$-machine model as well, where we use our fast connectivity algorithm (in place of
Thurimella's algorithm~\cite{Thurimella97} used in~\cite{GhaffariK13}). The time complexity
is dominated by the connectivity-testing procedure, and thus is $\tilde O(n/k^2)$ w.h.p.
\end{proof}

\subsection{Algorithms for Graph Verification Problems}

It is well known that graph connectivity is an important building block for several graph verification problems
(see, e.g., \cite{DasSarmaHKKNPPW12}). We now analyze some of such problems, formally defined, e.g.,
in Section~2.4 of \cite{DasSarmaHKKNPPW12}, in the $k$-machine model.

\begin{theorem}\label{thm:verification}
There exist algorithms for the $k$-machine model that solve  the following verification problems in $\tilde O(n/k^2)$
rounds with high probability: spanning connected subgraph, cycle containment, $e$-cycle containment, cut, $s$-$t$ connectivity,
edge on all paths, $s$-$t$ cut, bipartiteness.
\end{theorem}

\begin{proof}
We discuss each problem separately.
\begin{description}[itemsep=0.2ex,leftmargin=0cm]
\item[Cut verification:] remove the edges of the given cut from $G$, and then check whether the resulting graph is connected.
\item[$s$-$t$ connectivity verification:] run the connectivity algorithm and then verify whether $s$ and $t$ are in the same
connected component by checking whether they have the same label.
\item[Edge on all paths verification:] since $e$ lies on all paths between $u$ and $v$ iff $u$ and $v$ are disconnected
in $G \setminus \{e\}$, we can simply use the $s$-$t$ connectivity verification algorithm of previous point.
\item[$s$-$t$ cut verification:] to verify if a subgraph is an $s$-$t$ cut, simply verify $s$-$t$ connectivity of the graph after removing the edges of the subgraph.
\item[Bipartiteness verification:] use the connectivity algorithm and the reduction presented in Section~3.3 of \cite{AhnGM12a}.
\item[Spanning connected subgraph, cycle containment, and $e$-cycle containment verification:] these also follow from
the reductions given in~\cite{DasSarmaHKKNPPW12}.\qedhere
\end{description}
\end{proof}

\section{Lower Bounds for Verification Problems}\label{sec:LBs}

In this section we show that $\tilde \Omega(n/k^2)$ rounds is a fundamental lower bound for many graph
verification problems in the $k$-machine model. To this end we will use results from the classical theory
of communication complexity~\cite{KushilevitzN97}, a popular way to derive lower bounds in distributed
message-passing models~\cite{DasSarmaHKKNPPW12,Oshman14,PanduranganPS16}.

Even though many verification problems are known to satisfy a lower bound of $\tilde\Omega(D +\sqrt{n})$
in the classic distributed $\mathcal{CONGEST}$ model~\cite{DasSarmaHKKNPPW12}, the reduction
of~\cite{DasSarmaHKKNPPW12} encodes a $\Theta(\sqrt{n})$-instance of set disjointness, requiring at
least one node to receive $\tilde\Theta(\sqrt{n})$ information across a single short ``highway'' path or
via $\Theta(\sqrt{n})$ longer paths of length $\Theta(\sqrt{n})$. Moreover, we assume the
random vertex partition model, whereas the results of~\cite{DasSarmaHKKNPPW12}
assume a worst case distribution. Lastly, any pair of machines can communicate directly in the $k$-machine model,
thus breaking the $\Omega(D)$ bound for the $\mathcal{CONGEST}$ model.

Our complexity bounds follow from the communication complexity of $2$-player \emph{set disjointness} in
the \emph{random input partition model} (see~\cite{KlauckNPR15}). While in the standard model of communication
complexity there are $2$ players, Alice and Bob, and Alice (resp., Bob) receives an input vector $X$ (resp., $Y$)
of $b$ bits~\cite{KushilevitzN97}, in the random input partition model Alice receives $X$ and, in addition,
each bit of $Y$ has probability $1/2$ to be revealed to Alice. Bob's input is defined similarly with respect to $X$.
In the set disjointness problem, Alice and Bob must output $1$ if and only if there is
no index $i$ such that $X[i]=Y[i]=1$. The following result holds.

\begin{lemma}[{\cite[Lemma~3.2]{KlauckNPR15}}]\label{lem:setdisj}
For some constant $\epsilon >0$, every randomized communication protocol that solves set disjointness
in the random input partition model of $2$-party communication complexity with probability at least
$1- \epsilon$, requires $\Omega(b)$ bits.
\end{lemma}
Now we can show the main result of this section.

\begin{theorem}\label{thm:lowerbounds}
There exists a constant $\gamma>0$ such that any $\gamma$-error algorithm $\cA$ has round complexity
of $\tilde\Omega(n/k^2)$ on an $n$-node vertex graph of diameter $2$ in the $k$-machine model,
if $\cA$ solves any of the following problems: connectivity, spanning connected
subgraph, cycle containment, $e$-cycle containment, $s$-$t$-connectivity, cut, edge on all paths, and $s$-$t$-cut.
\end{theorem}

\begin{proof}
The high-level idea of the proof is similar to the simulation theorem of~\cite{DasSarmaHKKNPPW12}.
We present the argument for the spanning connected subgraph problem defined below.
The remaining problems can be reduced to the SCS problem using reductions similar to those in~\cite{DasSarmaHKKNPPW12}.

In the spanning connected subgraph (SCS) problem we are given a graph $G$ and a subgraph $H\subseteq G$
and we want to verify whether $H$ spans $G$ and is connected. We will show, through a reduction from $2$-party
set disjointness, that any algorithm for SCS in the $k$-machine model requires $\tilde \Omega(n/k^2)$ rounds.

Given an instance of the $2$-party set disjointness problem in the random partition model we will construct the following input graphs $G$ and $H$.
The nodes of $G$ consist of $2$ special nodes $s$ and $t$, and nodes $u_1,\dots,u_b$, $v_1,\dots,v_b$, for $b=(n-2)/2$. 
(For clarity of presentation, we assume that $(n-2)/2$ and $k/2$ are integers.)
The edges of $G$ consist of the edges $(s,t)$, $(u_i,v_i)$, $(s,u_i)$, $(v_i,t)$, for $1 \le i \le b$.

Let $\cM_A$ be the set of machines simulated by Alice, and let $\cM_B$ be the set of machines simulated by Bob, where $|\cM_A| = |\cM_B| = k/2$. 
First, Alice and Bob use shared randomness to choose the machines $M_X$ and $M_Y$ that receive the vertices $s$ and $t$.
If $M_X \ne M_Y$, then Alice assigns $t$ to a machine chosen randomly from $\cM_A$, and Bob assigns $s$ to a random machine in $\cM_B$.
Otherwise, if $M_X$ and $M_Y$ denote the same machine, Alice and Bob output $0$ and terminate the simulation.

The subgraph $H$ is determined by the disjointness input vectors $X$ and $Y$ as follows:
$H$ contains all nodes of $G$ and the edges $(u_i,v_i)$, $(s,t)$, $1 \le i \le b$.
Recall that, in the random partition model, $X$ and $Y$ are randomly distributed between Alice and Bob, but Alice knows all $X$ and Bob knows all of $Y$.
Hence, Alice and Bob mark the corresponding edges as being part of $H$ according to their respective input bits. 
That is, if Alice received $X[i]$ (i.e.\ Bob did \emph{not} receive $X[i]$), she assigns the node $u_i$ to a
random machine in $\cM_A$ and adds the edge $(s,u_i)$ to $H$ if and only if $X[i]=0$.
Similarly, the edge $(v_i,t)$ is added to $H$ if and only if $Y[i]=0$ (by either Alice or Bob depending on who receives $Y[i]$).
See Figure~\ref{fig:lowerbound}.
Note that, since $X$ and $Y$ were assigned according to the random input partition model, the resulting
distribution of vertices to machines adheres to the random vertex partition model.
Clearly, $H$ is an SCS if and only if $X$ and $Y$ are disjoint.

\begin{figure}[h]
    \begin{center}
       \includegraphics[width=0.42\columnwidth]{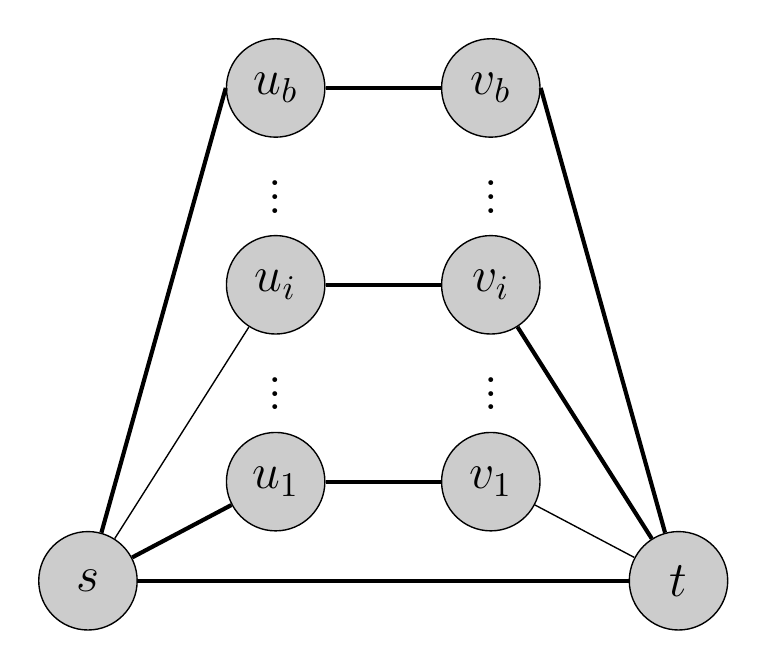}
       \caption{The graph construction for the spanning connected subgraph problem, given a set disjointness instance where $X[1]=0$, $Y[1]=1$, $X[i]=1$, $Y[i]=0$, and $X[b]=Y[b]=0$. The thick edges are the edges of subgraph $H$. The subgraph $H$ contains all edges $(u_i,v_i)$ ($1\le i \le b$) and $(s,t)$; 
       the remaining edges of $H$ are determined by the input vectors $X$ and $Y$ of the set disjointness instance.}
       \label{fig:lowerbound}
    \end{center}
\end{figure}

We describe the simulation from Alice's point of view (the simulation for Bob is similar):
Alice locally maintains a counter $r_A$, initialized to $1$, that represents the current round number. 
Then, she simulates the run of $\cA$ on each of her $k/2$ machines, yielding a set of $\ell$ messages
$m_1,\dots,m_\ell$ of $O(\polylog(n))$ bits each that need to be sent to Bob to simulate the algorithm on his machines in the next round.  
By construction, we have that $0 \le \ell \le \lceil k^2 /4\rceil$.
To send these messages in the (asynchronous) $2$-party random partition model of communication complexity,
Alice sends a message $\langle \ell, (M_1,m_1,M_2),\dots,(M_\ell,m_\ell,M_{\ell+1}) \rangle$ to Bob, where a
tuple $(M_i,m_i,M_{i+1})$ corresponds to a message $m_i$ generated by machine $M_i$ simulated by
Alice and destined to machine $M_{i+1}$ simulated at Bob. Upon receiving this message, Bob increases its
own round counter and then locally simulates the next round of his machines by delivering the messages to the appropriate machines.
Adding the source and destination fields to each message incurs an overhead of only $O(\log k) = O(\log n)$ bits,
hence the total communication generated by simulating a single round of $\cA$ is upper bounded by $\tilde O(k^2)$.
Therefore, if $\cA$ takes $T$ rounds to solve SCS in the $k$-machine model, then this gives
us an $O(T k^2 \polylog(n))$-bit communication complexity protocol for set disjointness in the random
partition model, as the communication between Alice and Bob is determined by the communication across
the $\Theta(k^2)$ links required for the simulation, each of which can carry $O(\polylog(n))$ bits per round.
Note that if $\cA$ errs with probability at most $\gamma$, then the simulation errs with probability at most
$\gamma+1/k$, where the extra $1/k$ term comes from the possibility that machines $M_X$ and $M_Y$ refer to the same machine.
For large enough $k$ and small enough $\gamma$ we have $\gamma+1/k < \epsilon$.
It follows that we need to simulate at least $T=\tilde\Omega(n/k^2)$ many rounds, since by Lemma~\ref{lem:setdisj} the
set disjointness problem requires $\Omega(b)$ bits in the random partition model, when the error is smaller than $\epsilon$.
\end{proof}

Interestingly, our lower bounds hold even for graphs of diameter $2$, which is in contrast
to the analogous results for the classic distributed $\mathcal{CONGEST}$ model assumed
in~\cite{DasSarmaHKKNPPW12}. We remark that the lower bound of connectivity verification was already shown in~\cite{KlauckNPR15}.

\section{Conclusions}

There are several natural directions for future work.
Our connectivity algorithm is randomized: it would be interesting to study the deterministic
complexity of graph connectivity in the $k$-machine model.
Specifically, does graph connectivity admit a $\tilde{O}(n/k^2)$ deterministic algorithm? 
Investigating higher-order connectivity, such as $2$-edge/vertex connectivity, is also an interesting
research direction. A general question motivated by the algorithms presented in this paper is
whether one can design algorithms that have superlinear scaling in $k$ for other fundamental
graph problems. Some recent results in this directions are in~\cite{PanduranganRS16}.

\vspace{0.4cm}
\paragraph{Acknowledgments}
The authors would like to thank Mohsen Ghaffari, Seth Gilbert, Andrew McGregor, Danupon Nanongkai, and Sriram V.\ Pemmaraju for helpful discussions.

\bibliographystyle{plain}
\bibliography{biblio}

\newpage
\appendix

\section{Omitted Proofs}\label{sec:omitted}

\subsection{Proof of Lemma~\ref{lem:depthDRR}}

\begin{proof}
Consider one phase of the algorithm, and suppose that during that phase there are $n$ components.
(In one phase there are $c \leq n$ components, thus setting $c = n$ gives a valid upper bound to
the height of each DRR tree in that phase.) Each component picks a random rank from $[0,1]$.
Thus, all ranks are distinct with high probability. If the target component's rank is higher, then the
source component connects to it, otherwise the source component becomes a root of a DRR tree.

Consider an arbitrary component of the graph, and consider the (unique) path $P$ starting form
the node that represents the component to the root of the tree that contains it. Let $|P|$ be the
number of nodes of $P$, and assign indexes to the $|P|$ nodes of $P$ according to their position
in the path from the selected node to the root. (See Figure~\ref{fig:DRR-tree}.) 

\begin{figure}[h]
    \begin{center}
       \includegraphics[width=0.7\columnwidth]{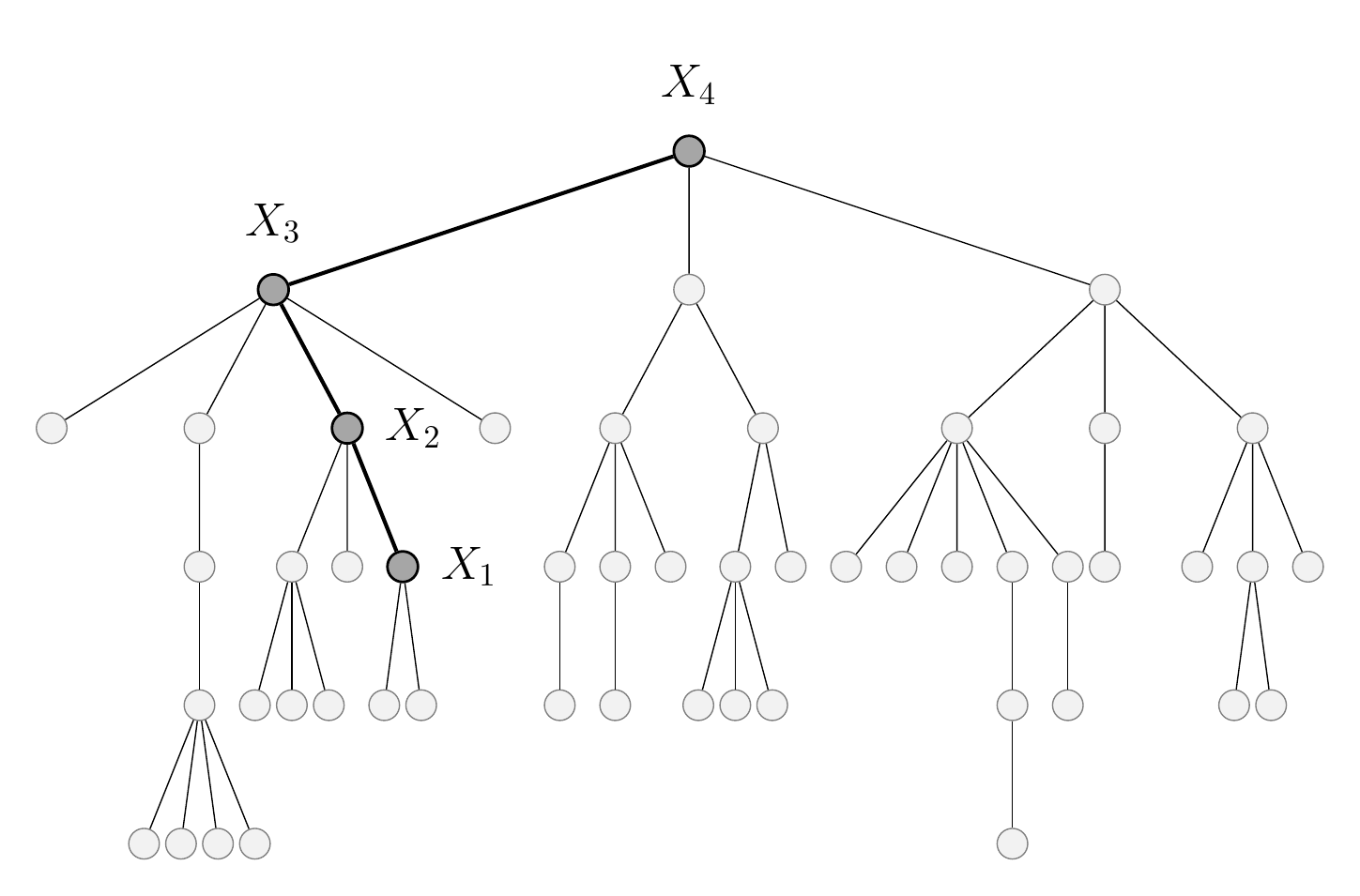}
       \caption{One DRR tree, and one path from one node to the root of the tree. Nodes of the path are
       labeled with the indicator variable associated to them, indexed by the position of the node in the path.}
       \label{fig:DRR-tree}
    \end{center}
\end{figure}

For each $i \in [|P|]$, define $X_i$ as the indicator variable that takes value $1$ if node $i$ is not
the root of $P$, and $0$ otherwise. Then, $X = \sum_{i=1}^{|P|} X_i$ is the length of the path $P$.
Because of the random choice for the outgoing edge made by components' parts, the outgoing edge
of each component is to a random (and distinct) component. This means that, for each $j \leq |P|$,
the ranks of the first $j$ nodes of the path form a set of $j$ random values in $[0,1]$. Hence,
the probability that a new random value in $[0,1]$ is higher than the rank of the $j$-th node of the
path is the probability that the new random value is higher than all the $j$ previously chosen random
values (that is, the probability its value is the highest among all the first $j$ values of the path),
and this probability is at most $1/(j+1)$. Thus, $\text{Pr}(X_i = 1) \leq 1/(i+1)$.
Hence, by the linearity of expectation, the expected height of a path in a tree produced by the
DRR procedure is
\begin{align*}
\text{E}[X] &= \sum_{i=1}^{|P|} \text{E}[X_i] \\
&\leq \sum_{i=1}^{n} \text{E}[X_i] \\
&= \sum_{i=1}^{n} \text{Pr}(X_i = 1) \\
&\leq \sum_{i=1}^{n} \frac{1}{i+1} \\
&\leq \log (n +1).
\end{align*}

Notice that the $X_i$'s are independent (but not identically distributed) random variables,
since the probability that the $i$-th smallest ranked node is not a root depends only on the
random neighbor that it picks, and is independent of the choices of the other nodes.
Thus, applying a standard Chernoff bound (see, e.g., \cite{MitzenmacherU05}) we have
\[
\text{Pr}(X \geq 6 \log (n+1)) \leq 2^{-6 \log (n+1)} = \frac{1}{(n+1)^6}.
\]
Applying the union bound over all the at most $n$ paths concludes the proof.
\end{proof}

\end{document}